\documentclass[sn-mathphys-num,iicol]{sn-jnl}

\usepackage{amsmath,amssymb,amsthm,bm,graphicx}
\theoremstyle{plain}
\newtheorem{theorem}{Theorem}
\newtheorem{proposition}{Proposition}
\newtheorem{corollary}{Corollary}

\theoremstyle{remark}
\newtheorem*{remark}{Remark}
\newcommand{\Mm}{\mathcal{M}}

\begin{document}

\title[Certified breathing stability regions]{Certified Breathing Stability Regions in Nonlinear Dynamical Systems: Composite Lyapunov Certificates, M-Matrix Conditions, and a Resilience--Fragility Correspondence}

\author[1]{\fnm{Mari\'an} \sur{Me\v{s}ter}}\email{marian.mester@tuke.sk}
\affil[1]{\orgdiv{Department of Electric Power Engineering}, \orgname{FEI, Technical University of Ko\v{s}ice}, \city{Ko\v{s}ice}, \country{Slovakia}}

\abstract{We develop a unified, certified lower bound on the time-to-boundary margin $\Mm$ for transient stability of interconnected dissipative systems under slow parameter drift. The companion work establishes $\Mm$ as the first-passage time of the joint state--parameter motion to the synchronism boundary and proves $\Mm=\mathrm{CCT}$ exactly on the one-machine--infinite-bus reduction, while leaving the multimachine certified margin open. Here a composite (mixed-region) Lyapunov function, formed by absorbing the restoring intra-group coupling into group energy functions and treating only the residual cross-cut coupling through the comparison principle, yields a positively invariant inner estimate of the region of attraction whenever an associated test matrix is a nonsingular M-matrix. The certified region \emph{breathes} with the drift: its size is governed by a single critical synchronising stiffness $k_c(\lambda)$, and as $k_c\to0$ at the boundary the region breathes shut and the certified margin $\Mm_{\mathrm{low}}\le\Mm_{\mathrm{true}}$ vanishes. We give a nonlinear sector form of the construction, a domain-neutral resilience--fragility reading in which the coupling that certifies order is the one whose growth certifies collapse, and a constructive control corollary establishing a sharp dichotomy between damping injection and structural action. The mechanism is demonstrated identically on the WSCC nine-bus power system and on an inertial Kuramoto network, whose normalised breathing curves collapse, to leading order, onto a single profile. We present this collapse as numerical evidence for a conjectured universal form; a normal-form proof is identified as the precise open step.}


\maketitle

\section{Introduction}\label{sec:intro}
Power systems are being pushed closer to their stability limits by the replacement of synchronous generation with converter-interfaced renewables, which erodes inertia and shrinks the corrective window left after a disturbance [8,9]. In this low-inertia regime the operationally decisive question is no longer only whether the system is secure at the present operating point, but how long it will remain so as that point drifts: the quantity of interest is a time to a boundary, not a binary verdict. A recent European disturbance illustrates the operational stakes: in the Czech grid incident of 4 July 2025 an ordinary single-line contingency cascaded to an island collapse, with the corrective actions arriving too slowly and the cascading risk present in routine $N-1$ results yet not surfaced to the operator in real time [49].

The classical index answering a fixed-point version of this question is the critical clearing time (CCT). The companion paper [1] generalises CCT to a time-to-boundary margin $\Mm$, the first time the joint state--parameter trajectory reaches the survival boundary $\Sigma=\partial A(\delta^*)$ under slow drift, and proves $\Mm=\mathrm{CCT}$ exactly on the one-machine--infinite-bus (OMIB) reduction. It also delimits the multimachine case honestly: the direct energy method is non-conservative there, the controlling unstable equilibrium being the binding obstruction to a \emph{certified} multimachine margin.

This places the present work within a long line of transient-stability certificates, usefully separated into five strands. \emph{(i) Direct and energy-function methods} estimate the region of attraction through a scalar energy and a controlling-UEP or PEBS threshold [10,11]. \emph{(ii) Lyapunov-family and convex certificates} relax the energy function into optimised or sum-of-squares Lyapunov functions, enlarging the certified set [12,13]. \emph{(iii) Vector-Lyapunov and compositional methods} aggregate per-subsystem certificates through an M-matrix or comparison condition [14,2,4]; this is the tradition the present construction belongs to. \emph{(iv) Synchronisation and basin-stability analyses} characterise the in-phase manifold and the size of its basin, often through Kuramoto reductions [15,16,7]. \emph{(v) Critical-transition and early-warning} viewpoints read the approach to a bifurcation as a slow loss of restoring stiffness, developed in Section~\ref{sec:duality}.

The contribution of this paper is the complementary half of the companion's open problem. Where the energy margin of [1] \emph{overestimates} $\Mm$ -- it certifies a level the true motion can exceed -- the present certificate \emph{underestimates} it: a guaranteed positively invariant inner estimate of the survival set yields a certified lower bound $\Mm_{\mathrm{low}}\le\Mm_{\mathrm{true}}$ that follows the operating point as it drifts. The two bound $\Mm_{\mathrm{true}}$ from opposite sides.

Vector-Lyapunov M-matrix certificates for multimachine stability are themselves classical [14,2], and we separate our contribution cleanly from that inherited base. What is new is fourfold: (i) a $\lambda$-dependent test matrix tied to $\Mm$ as a first-passage time, not to asymptotic stability alone; (ii) a nonlinear sector bound exhibiting a \emph{double squeeze} under stress (Proposition~\ref{prop:sector}), invisible on the OMIB; (iii) a single critical synchronising stiffness $k_c(\lambda)$ governing the certified region size -- the \emph{breathing} law $\rho_{\mathrm{cert}}\propto k_c$ (Theorem~\ref{thm:main}); and (iv) a constructive control dichotomy (Corollary~\ref{cor:control}) with a domain-neutral resilience--fragility reading. The classical vector-Lyapunov certificates [14,2,3] establish asymptotic stability at \emph{fixed} parameters; none ties the M-matrix gate to a first-passage time under drift, nor extracts a single-stiffness breathing law. Closest among prior compositional constructions is [4] (whose Theorem~2 was subsequently corrected~\cite{caliskan_corr}); we differ in the drift-aware, time-to-boundary target and the breathing law. Indeed, the very multimachine composition whose proof required correction in~\cite{caliskan_corr} is here handled cleanly through the M-matrix gate on the composite certificate.

The construction is stated for a class of interconnected dissipative systems; the multimachine power system is one instance, not the scope. The duality of [1] -- the coupling that builds synchronising order is the channel of collapse -- reappears in domain-neutral form: $k_c$ \emph{is} the synchronising stiffness (order), and its erosion to zero \emph{is} the collapse (fragility). We demonstrate the identical mechanism on a power system and on an inertial Kuramoto network, where the normalised breathing curves collapse, to leading order, onto one shape -- numerical evidence for a structural homology (rigour level~2), with the normal-form proof left open (Section~\ref{sec:gen}).

\section{Preliminaries}\label{sec:prelim}
We adopt the time-to-boundary margin $\Mm$ and the survival boundary $\Sigma=\partial A(\delta^*)$ from [1] without re-deriving them. Here $\Sigma=\partial A(\delta^*)$ is the basin-boundary (loss-of-synchronism) instance of the general survival boundary of [1], which takes whichever of the protection-time limit or the saddle-node is reached first. To keep the paper self-contained: for a system $\dot z=F(z,\lambda)$ with stable equilibrium $\delta^*(\lambda)$, region of attraction $A(\delta^*(\lambda))$, and slow drift $\dot\lambda=\rho$, the survival boundary is $\Sigma=\partial A(\delta^*)$ and $\Mm(z_0,\lambda_0)$ is the first time the joint trajectory $(z(t),\lambda(t))$ started at $(z_0,\lambda_0)$ reaches $\Sigma$. On the OMIB reduction $\Mm$ equals the critical clearing time exactly [1], certified there on a published benchmark; the energy-margin estimate $\Mm_{\mathrm{energy}}$ of [1] is an upper estimate, non-conservative in the multimachine case. We inherit from [1] the quasi-static drift hypothesis $\rho\,T_{\mathrm{sw}}\ll1$ (drift slow on the swing time scale); we return in Section~\ref{sec:disc} to its validity in low-inertia systems, where fast converter dynamics make the separation of time scales least automatic and the margin is then to be read per swing epoch between fast events.

The certificate is assembled by the \emph{comparison principle}. If a vector $w(t)=(w_1,\dots,w_r)^{\!\top}$ of nonnegative scalar functions satisfies $\dot w\le f(w)$ componentwise for a quasi-monotone (cooperative) $f$, then $w(t)\le u(t)$ for the solution $u$ of the comparison system $\dot u=f(u)$ with $u(0)\ge w(0)$ [17]. Aggregating a vector of subsystem Lyapunov functions through such a comparison system is the vector-Lyapunov method [18,19,2,3], whose stability test reduces to a property of the (constant) comparison matrix.

That property is the M-matrix condition. A real matrix $W$ with nonpositive off-diagonal entries ($W_{ij}\le0$, $i\ne j$) is a nonsingular \emph{M-matrix} if any of the following equivalent conditions holds: its leading principal minors are positive; there exist weights $d_i>0$ with $d_iW_{ii}>\sum_{j\ne i}d_j|W_{ij}|$ (generalised diagonal dominance); or $-W$ is Hurwitz [20,5]. The last form links the test directly to the comparison flow: a Metzler matrix $B$ (nonnegative off-diagonal) generates a cooperative, order-preserving flow, and ``$-B$ is a nonsingular M-matrix'' is exactly the condition that this flow is asymptotically stable [21,14]. This equivalence is the engine of Theorem~\ref{thm:main}.

\section{Composite (mixed-region) Lyapunov function}\label{sec:composite}
A per-machine decomposition fails on a strongly coupled system, and the failure is instructive. If every inter-machine coupling is treated as a destabilising disturbance, the certificate is vacuous even at a secure operating point (Section~\ref{sec:num}), because the synchronising coupling is \emph{restoring} near the in-phase equilibrium, not disturbing; discarding it discards the dominant stabilising physics -- the known limitation of the vector-Lyapunov programme for strong coupling [2].

The composite construction partitions the machines into groups $\{G_a\}$ across \emph{weak} cuts. The natural partition is the slow-coherent grouping, in which machines within a group swing together and the cuts are the weak inter-area links [23,24]; this is precisely the cutset structure that topology-preserving stability models exploit, with critical cutsets governing transient stability [25]. Within a group the restoring coupling is absorbed into a group energy function $v_a$; only the residual coupling across cuts is treated by comparison. The ``mixed region'' is this split: an energy-type certificate inside each group, a comparison-type certificate between them. For a damped group we use the skewed (Chetaev) form [22], which makes the decay strict rather than merely non-increasing,
\begin{equation}
v_a=\tfrac12\,\omega_a^{\!\top}M_a\,\omega_a+W_a(\theta_a)+\beta_a\,\omega_a^{\!\top}M_a\,\theta_a,
\label{eq:chetaev}
\end{equation}
where $M_a$ is the group inertia, $\omega_a$ the speed deviation, $\theta_a$ the angle deviation from $\delta^*$, $W_a$ the group potential well, and $\beta_a>0$ a small skew parameter. Positivity $v_a\succ0$ requires $\beta_a^2<\lambda_{\min}(M_a^{-1}\nabla^2W_a)$, and strict decay $-\dot v_a\succ0$ requires additionally the \emph{damping-admissibility} bound $\beta_a<\min\{D_a/M_a,\;k_aD_a/(k_aM_a+\tfrac14D_a^2)\}$ on the binding mode ($k_a=\lambda_{\min}(\nabla^2W_a)$): the skew must inject dissipation without out-tilting the damping, a constraint absent from the energy form. The cross term injects dissipation along the angle direction, where velocity damping alone gives none; the limit $\beta_a\to0$ recovers the energy method of [1] continuously.

\section{Main result: the certified breathing region}\label{sec:main}
Let the joint motion be $\dot z=F(z,\lambda)$, $z=(\theta,\omega)$, with slow drift $\dot\lambda=\rho$. In $w_a=\sqrt{v_a}$ the bilinear cross-cut coupling becomes a genuine rate. The certificate below is an inner sublevel-set estimate of the region of attraction -- the classical Lyapunov route to that region [26,27] -- here made $\lambda$-dependent and tied to $\Mm$.

\begin{proposition}[Nonlinear cross-cut rate]\label{prop:sector}
For a cut with lossless coupling $g_{ab}(\phi)=E_aE_bB_{ab}[\sin\delta_{ab}-\sin\delta_{ab}^*]$, $\phi=\delta_{ab}-\delta_{ab}^*$, and excursions $|\phi|\le\bar\phi$,
\begin{equation}\label{eq:sector}
\begin{aligned}
|g_{ab}(\phi)|&\le\Gamma_{ab}|\phi|,\\
\Gamma_{ab}&=E_aE_bB_{ab}\Big[|\cos\delta_{ab}^*|+|\sin\delta_{ab}^*|\,\tfrac{1-\cos\bar\phi}{\bar\phi}\Big].
\end{aligned}
\end{equation}
\end{proposition}
\begin{proof}
Writing $\sin\delta_{ab}=\sin\delta_{ab}^*\cos\phi+\cos\delta_{ab}^*\sin\phi$ gives
\begin{equation*}
\sin\delta_{ab}-\sin\delta_{ab}^*=\cos\delta_{ab}^*\sin\phi-\sin\delta_{ab}^*(1-\cos\phi),
\end{equation*}
hence $|g_{ab}|\le E_aE_bB_{ab}\big[|\cos\delta_{ab}^*|\,|\sin\phi|+|\sin\delta_{ab}^*|(1-\cos\phi)\big]$. On $(0,\pi/2]$ the map $\phi\mapsto(1-\cos\phi)/\phi$ is increasing -- the numerator of its derivative, $\phi\sin\phi-(1-\cos\phi)$, vanishes at $0$ and has derivative $\phi\cos\phi\ge0$ there -- so $1-\cos\phi\le\frac{1-\cos\bar\phi}{\bar\phi}|\phi|$ for $|\phi|\le\bar\phi\le\pi/2$. With $|\sin\phi|\le|\phi|$ the bound \eqref{eq:sector} follows.
\end{proof}
\noindent Here $\delta_{ab}^*(\lambda)$ is the across-cut equilibrium angle; as $\bar\phi\to0$, $\Gamma_{ab}\to E_aE_bB_{ab}|\cos\delta_{ab}^*|=K_{ab}$, the linear synchronising coefficient. The two terms are a \emph{double squeeze}: under stress $\delta_{ab}^*$ rises, so $|\cos\delta_{ab}^*|$ falls (the coupling weakens) while $|\sin\delta_{ab}^*|$ grows (nonlinear inflation) -- both act against the margin at once, an effect invisible on the OMIB.

Define the test matrix $W(\lambda,\bar\phi)$ by $W_{aa}=c_a(\lambda)/2$, $W_{ab}=-\eta_{ab}$, $\eta_{ab}=\Gamma_{ab}/\sqrt{p_a k_b}$, with $p_a=\lambda_{\min}(P_a)$ and $k_b$ the group stiffness.

\begin{theorem}[Certified breathing region]\label{thm:main}
Suppose \eqref{eq:chetaev} gives strict group decay $c_a(\lambda)>0$ and Proposition~\ref{prop:sector} the cross-cut rates $\eta_{ab}$. If $W(\lambda,\bar\phi)$ is a nonsingular M-matrix, there exist $d_a>0$ such that $V_\lambda=\sum_a d_a v_a$ is a strict Lyapunov function and $\Omega(\lambda)=\{V_\lambda\le\kappa(\lambda)\}$ is a positively invariant inner estimate of $A(\delta^*(\lambda))$ whose angular radius satisfies $\rho_{\mathrm{cert}}(\lambda)\propto k_c(\lambda)$, with $k_c(\lambda)=\min_a\lambda_{\min}(\nabla^2W_a)$ over the binding cut. Consequently $\Mm_{\mathrm{low}}(\lambda)\le\Mm_{\mathrm{true}}(\lambda)$, and as $k_c(\lambda)\to0$ at $\lambda\to\lambda^*$, $\Omega(\lambda)$ breathes shut and $\Mm_{\mathrm{low}}\to0$.
\end{theorem}
\begin{proof}
\emph{Step 1 (group decay).} Linearise the isolated group $a$ about its equilibrium: with state $x_a=(\theta_a,\omega_a)$, $\dot x_a=A_ax_a$, $A_a=\left[\begin{smallmatrix}0&I\\-M_a^{-1}K_a&-M_a^{-1}D_a\end{smallmatrix}\right]$, where $K_a$ is the intra-group block of the reduced power-angle Jacobian and $D_a\succ0$ the damping (taken uniform, $D_a=2.0$~p.u., in the computations below). For a damped restoring group $A_a$ is Hurwitz, so the Lyapunov equation $A_a^{\!\top}P_a+P_aA_a=-I$ has a unique $P_a\succ0$, and the skewed form \eqref{eq:chetaev} is its quadratic realisation $v_a=x_a^{\!\top}P_ax_a$. Then $\dot v_a=-\lVert x_a\rVert^2\le-v_a/\lambda_{\max}(P_a)$, so in $w_a=\sqrt{v_a}$ the decay is strict, $\dot w_a\le-c_aw_a$ with $c_a=1/\!\left(2\lambda_{\max}(P_a)\right)>0$; the skew term is what makes $-\dot v_a$ positive definite in $(\theta_a,\omega_a)$, which velocity damping alone does not. For the WSCC base case ($\lambda=1$), with group $A=\{\mathrm{gen}\,2\}$ (the critical machine) and $B=\{\mathrm{gen}\,1,\mathrm{gen}\,3\}$, $A_A$ has eigenvalues $\{-1.35,-57.6\}$ and $\lambda_{\max}(P_A)=1.05$, giving $c_A=0.48~\mathrm{s}^{-1}$; $A_B$ has eigenvalues $\{-0.72,-2.25,-14.2,-124\}$ and $\lambda_{\max}(P_B)=1.28$, giving $c_B=0.39~\mathrm{s}^{-1}$ -- both strictly positive. The cross-cut coupling contributes $\sum_{b\ne a}(\omega_a+\beta_a\theta_a)^{\!\top}g_{ab}$, treated in Step~2.

\emph{Step 2 (rate form).} Put $w_a=\sqrt{v_a}$. With $|\omega_a+\beta_a\theta_a|\le\sqrt{2v_a/p_a}=\sqrt2\,w_a/\sqrt{p_a}$ and, by Proposition~\ref{prop:sector}, $|g_{ab}|\le\Gamma_{ab}(|\theta_a|+|\theta_b|)\le\Gamma_{ab}\big(\sqrt{2v_a/k_a}+\sqrt{2v_b/k_b}\big)$, the chain rule $\dot w_a=\dot v_a/(2w_a)$ gives $\dot w_a\le-(c_a/2)w_a+\sum_{b\ne a}\eta_{ab}w_b$ after folding the $|\theta_a|$ self-term into the diagonal, with $\eta_{ab}=\Gamma_{ab}/\sqrt{p_ak_b}$. The residual coupling enters only the velocity channel, so $\eta_{ab}$ is the metric-scaled norm of that single block, tighter than a norm of the full inter-group operator.

\emph{Step 3 (comparison).} Hence $\dot w\le Bw$ componentwise, with $B_{aa}=-c_a/2$ and $B_{ab}=\eta_{ab}\ge0$ ($b\ne a$). $B$ is Metzler, so $\dot u=Bu$ is cooperative and $w(t)\le u(t)$ whenever $w(0)\le u(0)$.

\emph{Step 4 (M-matrix $\Rightarrow$ Lyapunov).} $W=-B$ has the M-matrix sign pattern; $W$ a nonsingular M-matrix is equivalent to $B$ Hurwitz and to the existence of $d\succ0$ with $d^{\!\top}B\le-\varepsilon\,d^{\!\top}$, $\varepsilon>0$ [5]. Then $L(w)=d^{\!\top}w$ obeys $\dot L\le d^{\!\top}Bw\le-\varepsilon L$, so each sublevel set $\Omega(\lambda)=\{L\le\kappa\}$ is positively invariant, and $L>0$ off $\delta^*$ makes $\Omega\subseteq A(\delta^*)$ an inner estimate.

\emph{Step 5 (lower bound).} Since $\Omega\subseteq A$, the joint trajectory leaves $\Omega$ no later than it leaves the survival set, so the first-passage time of $\partial\Omega$ lower-bounds $\Mm$.

\emph{Step 6 (breathing).} The admissible $\kappa$ is bounded by $L$ at the controlling saddle of the binding cut -- the boundary-controlling unstable equilibrium whose energy sets the threshold in direct methods [28,29,30]. For a pendular cut the saddle barrier scales as $k_c^{3}$ near the fold (saddle-node scaling), while the basin angular half-width $\delta_u^*-\delta_s^*=\pi-2\delta_s^*\propto k_c$; the induced inner radius therefore satisfies $\rho_{\mathrm{cert}}\propto k_c$. As $k_c(\lambda)\to0$ at $\lambda\to\lambda^*$, $\Omega(\lambda)$ shrinks to $\{\delta^*\}$ and $\Mm_{\mathrm{low}}\to0$. The self-consistent $\kappa$ solves $\bar\phi(\kappa)\mapsto\Gamma(\bar\phi)\mapsto$ M-matrix feasibility, structurally the boundary-level fixed point underlying the OMIB first-passage construction of [1, Prop.~1].
\end{proof}

\begin{proposition}[Bracketing]\label{prop:bracket}
Under the hypotheses of [1] and Theorem~\ref{thm:main}, for each admissible $\lambda$,
\begin{equation*}
\Mm_{\mathrm{low}}(\lambda)\;\le\;\Mm_{\mathrm{true}}(\lambda)\;\le\;\Mm_{\mathrm{energy}}(\lambda),
\end{equation*}
where $\Mm_{\mathrm{low}}$ is the certified first-passage time of $\partial\Omega$ (Theorem~\ref{thm:main}), $\Mm_{\mathrm{energy}}$ the energy-margin estimate of [1], and $\Mm_{\mathrm{true}}$ the true first-passage time of $\Sigma$.
\end{proposition}
The lower bound is proved above (Step~5); the upper bound is the controlling-UEP energy-margin estimate of [1], which overestimates the critical energy -- and hence $\Mm$ -- in the multimachine case [1, Sec.~IV]. The lower bound is unconditional. The upper bound is that specific energy estimate, not a bare static-fold time; the latter can itself be exceeded under finite drift by dynamic ride-through (Section~\ref{sec:disc}), so only $\Mm_{\mathrm{low}}$ is guaranteed.

\begin{remark}[M-matrix is the existence gate, $k_c$ is the breathing driver]
The M-matrix condition certifies that $\Omega$ is a \emph{valid} inner estimate; it does not drive the breathing. The size is governed by $k_c(\lambda)$. On well-conditioned systems the M-matrix holds with margin throughout the breathing; on stiff, strongly coupled power systems the crude comparison bound is loose (Section~\ref{sec:num}), but $k_c(\lambda)$ remains well-defined and tight. Certification is thus non-vacuous in the weak-cut regime that is the construction's standing hypothesis (groups separated by genuinely weak cuts, Section~\ref{sec:composite}); outside it the theorem supplies the breathing \emph{mechanism} and the $k_c$ law rather than a tight region, and the operational content is carried by $k_c(\lambda)$, not by the area of $\Omega$.
\end{remark}

\section{The resilience--fragility duality}\label{sec:duality}
Theorem~\ref{thm:main} states the duality without metaphor. The certified region lives or dies by $k_c(\lambda)$, the across-cut synchronising stiffness: the very quantity that holds the groups in step (order) is the one whose erosion to zero is the collapse (fragility). One quantity, two faces. Where [1] places this multimachine duality at a plausible analogy (rigour level~3, uncertified beyond the OMIB pillar), the composite certificate here raises it to a structural homology (rigour level~2) in the weak-cut regime. The reading is sharper, and less symmetric, than on the OMIB: there it was coupling against coupling; here it is the critical synchronising stiffness against the drift that consumes it.

This reading connects to the critical-transitions programme. The drift to $\lambda^*$ at which $k_c\to0$ is a saddle-node of the across-cut equilibrium -- in the fast--slow classification of critical transitions, the generic codimension-one fold reached under slow parameter motion [32]; this identification is exact (rigour level~1). Its dynamical signature is critical slowing down: as the restoring stiffness vanishes the slowest recovery time diverges, observable as rising autocorrelation and variance [31], and used in power systems as a model-free proximity indicator to instability and voltage collapse, tied analytically to the saddle-node of the stochastic machine model [33,34]. The present margin is the deterministic, certified counterpart of that statistical precursor: where slowing down detects $k_c\to0$ from data, $\Mm_{\mathrm{low}}\propto k_c$ bounds the time that remains, from the model. This counterpart relation is a structural homology (rigour level~2) -- one vanishing stiffness organises both the certified inner radius and the diverging recovery time -- and it extends to the basin geometry, whose topology-induced fragilities such as dead-end nodes erode the order the coupling once supplied [35].

\section{Control reading: re-opening the breathing region}\label{sec:control}
The margin of group $a$ is $m_a(\lambda)=c_a(\lambda)/2-\sum_{b\ne a}\eta_{ab}(\lambda)$. Two levers re-open a closed region, and they obey a sharp dichotomy set by the natural frequency of the binding mode. This is the control face of the same certificate: a control vector-Lyapunov function turns the M-matrix margin condition into decentralised feedback synthesis [36].

\emph{Damping authority is bounded.} The decay rate of a damped second-order mode $M_a\ddot\theta=-D_a\dot\theta-k_a\theta$ is $D_a/(2M_a)$ while underdamped and decreases as $k_a/D_a$ once overdamped; it is maximised at critical damping $D_a^\star=2\sqrt{M_a k_a}$, with ceiling
\begin{equation}
c_a^{\max}=\sqrt{k_a/M_a}=\omega_{n,a},
\label{eq:ceiling}
\end{equation}
the modal natural frequency. Over-damping ($D_a>D_a^\star$) is counterproductive. Hence damping injection can hold the region open only up to this ceiling -- the classical damping-torque limit [37].

\begin{corollary}[Control dichotomy]\label{cor:control}
Let group $a$ have a violated margin $m_a(\lambda)<0$ at loading $\lambda$. Then:
\begin{itemize}
\item[(i)] if $\sum_{b\ne a}\eta_{ab}(\lambda)<\tfrac12\sqrt{k_a(\lambda)/M_a}$, there exists a damping injection $D_a\le D_a^\star$ that restores $m_a(\lambda)\ge0$;
\item[(ii)] if $\sum_{b\ne a}\eta_{ab}(\lambda)\ge\tfrac12\sqrt{k_a(\lambda)/M_a}$, no damping restores the margin, and only structural action -- selective tripping or redispatch that raises $k_a$ or lowers $\eta_{ab}$ -- re-opens $\Omega(\lambda)$.
\end{itemize}
As $\lambda\to\lambda^*$, $k_a\to0$ forces case~(ii): the lever at the fold is structural, recovering the selective-generator-tripping control of [1].
\end{corollary}
\begin{proof}
The roots of $M_a\mu^2+D_a\mu+k_a=0$ have real part $-D_a/(2M_a)$ when $D_a^2<4M_ak_a$, and slowest magnitude $(D_a-\sqrt{D_a^2-4M_ak_a})/(2M_a)\to k_a/D_a$ as $D_a\to\infty$ when overdamped; the decay rate thus increases on $[0,D_a^\star]$ and decreases beyond, with maximum $\sqrt{k_a/M_a}$ at $D_a^\star=2\sqrt{M_ak_a}$. The certified diagonal $c_a/2$ inherits this ceiling up to the Lyapunov constant, so $\sup_{D_a}m_a(\lambda)=\tfrac12\sqrt{k_a/M_a}-\sum_{b\ne a}\eta_{ab}$. If positive, continuity of $c_a(D_a)$ from $0$ to the ceiling yields a $D_a\le D_a^\star$ with $m_a=0$ (case i); if $\le0$, no $D_a$ restores the margin and only a change of $k_a$ or $\eta_{ab}$ can (case ii). As $\lambda\to\lambda^*$, $k_a\to0$ drives the ceiling below $\sum_b\eta_{ab}$, forcing case (ii).
\end{proof}

In a low-inertia system the damping lever of case~(i) is realised as converter virtual damping, a freely tunable control parameter rather than a fixed machine constant [38]; this widens the practical reach of case~(i), but it does not move the fold, where $k_a\to0$ makes case~(ii) unavoidable.

The operational reading mirrors [1]: control buys lead time by flattening $\mathrm{d}k_c/\mathrm{d}\lambda$, and the scarcer the corrective flexibility, the earlier it must trigger. Critically, the \emph{authority} of the cheap lever erodes with the margin: on the WSCC critical machine (Section~\ref{sec:num}) the damping ceiling $\omega_{n,c}=\sqrt{k_c/M_c}$ falls from $8.8$ to $4.9~\mathrm{rad/s}$ ($56\%$ of base) as $\lambda\to\lambda^*$ (Fig.~\ref{fig:ctrl}), so damping-only defence weakens precisely where it is most needed, and the dichotomy tips toward structural action. The control reading thus gives operators a decision rule -- damping problem versus topology problem -- read from $k_c(\lambda)$ alone.

\begin{figure}[tbp]\centering\includegraphics[width=.72\linewidth]{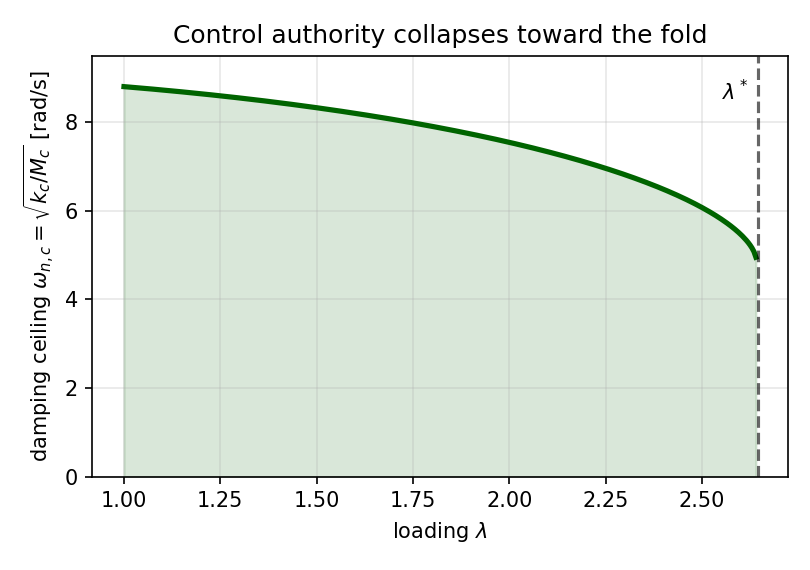}
\caption{Control authority collapses toward the fold: the damping ceiling $\omega_{n,c}=\sqrt{k_c/M_c}$ for the WSCC critical machine falls to $56\%$ of base as $\lambda\to\lambda^*$, tipping Corollary~\ref{cor:control} toward structural action.}\label{fig:ctrl}\end{figure}

\section{Numerical demonstration: two instances, one mechanism}\label{sec:num}
\paragraph{(i) WSCC nine-bus power system.}
We use the WSCC three-machine nine-bus system [6]; the pre-fault mechanical powers $0.716,1.63,0.85$~p.u.\ match the reference dispatch. Kron-reduced to the three internal nodes, with a uniform transfer stress $P_{m,i}=\lambda P^0_{m,i}$ re-solved at each $\lambda$ with a standard power-flow solver [43]; the saddle-node, the generic fold of the parameterised power-flow model [42], is at $\lambda^*\approx2.65$. The across-cut synchronising stiffness of the critical machine (gen~2, $H{=}6.4$) erodes monotonically from $k_c=2.63$ at $\lambda{=}1$ to $0.83$ at $\lambda{=}2.64$ -- to $32\%$ of base -- while the two stiff machines barely move (Fig.~\ref{fig:breath}a). Since the certified angular radius scales as $k_c$, the region breathes to $32\%$ of base (Fig.~\ref{fig:breath}b). These quantities are read directly from the reduced power-angle Jacobian; they involve no free constants.

\paragraph{(ii) Inertial Kuramoto network.}
Two communities of homogeneous second-order (inertial) Kuramoto oscillators [39,40,41] ($m{=}1$, damping $d$), strong intra- and weak inter-community coupling $K$, with a drift $\lambda$ scaling the inter-community power mismatch $\Delta P=\lambda\,\Delta P_0$. The across-cut equilibrium $\delta^*=\arcsin(\Delta P/K)$ rises to $\pi/2$, the cut stiffness $k_c=K\cos\delta^*$ erodes from $0.95$ to $0$ at $\lambda^*=K/\Delta P_0$, and the region breathes shut. No power-system specifics enter; the same Theorem~\ref{thm:main} applies verbatim.

\paragraph{Universality.}
Normalised by the distance to boundary $(\lambda-\lambda_0)/(\lambda^*-\lambda_0)$, the breathing curves of five instances -- three anchored two-area systems (varying inertia, dispatch, tie), the inertial Kuramoto pair, and the WSCC nine-bus -- collapse, to leading order, onto one profile (Fig.~\ref{fig:univ}a) and track the universal saddle-node exponent $\tfrac12$ on a log--log scale (Fig.~\ref{fig:univ}b). The clean saddle-node instances coincide to within $0.5\%$ near the boundary; the WSCC curve, whose power-flow fold is not a pure saddle-node in these coordinates, is the mild outlier -- consistent with a leading-order, not exact, collapse.

\paragraph{(iii) A certifying weak-cut instance.}
Instances (i)--(ii) exhibit the breathing driver $k_c(\lambda)$ but, with one machine per group, do not meet the weak-cut hypothesis, so the gate is conservative on them. To exhibit Theorem~\ref{thm:main} \emph{certifying}, we use two machines anchored to an infinite bus (well-conditioned restoring, no intra-group zero mode) and joined by a weak tie of strength $B_{\mathrm{tie}}$, with $\lambda$ scaling the transfer. With the damping-admissible skew weight (Section~\ref{sec:composite}) and the velocity-channel cross-coefficient (Step~2), the composite test matrix $W(\lambda)$ is a nonsingular M-matrix over $\lambda\in[1,\lambda_M)$, $\lambda_M=4.04$, while the true fold is $\lambda^*=4.12$: the certified region is non-empty over a genuine loading range and breathes shut \emph{conservatively}, before the saddle-node (Fig.~\ref{fig:cert}a). Certification holds for genuinely weak cuts, $B_{\mathrm{tie}}/B_{\mathrm{inf}}\lesssim0.09$, with a clean threshold (Fig.~\ref{fig:cert}b). Under slow drift the certified first-passage lower-bounds the simulated one, $\Mm_{\mathrm{low}}\le\Mm_{\mathrm{true}}$, across drift rates (Fig.~\ref{fig:cert}c).

\paragraph{Scope of the M-matrix gate.}
The gate certifies where its hypothesis holds -- anchored groups separated by a genuinely weak cut, instance (iii). With one machine per group and a strong cut, as in the nine-bus and the Kuramoto pair, the row-dominance margin is negative and the gate is conservative; there the operative result is the parameter-free $k_c(\lambda)$, the M-matrix supplying the existence gate where the weak-cut hypothesis is met. The conservatism is the cut strength, not the bound: the same construction certifies once the cut is weak (Fig.~\ref{fig:cert}b). This is the partial, weak-cut-conditional resolution of the certified-multimachine-margin problem posed as the principal open question of [1]; we do not claim closure beyond the weak-cut regime.

\begin{figure}[tbp]\centering\includegraphics[width=\linewidth]{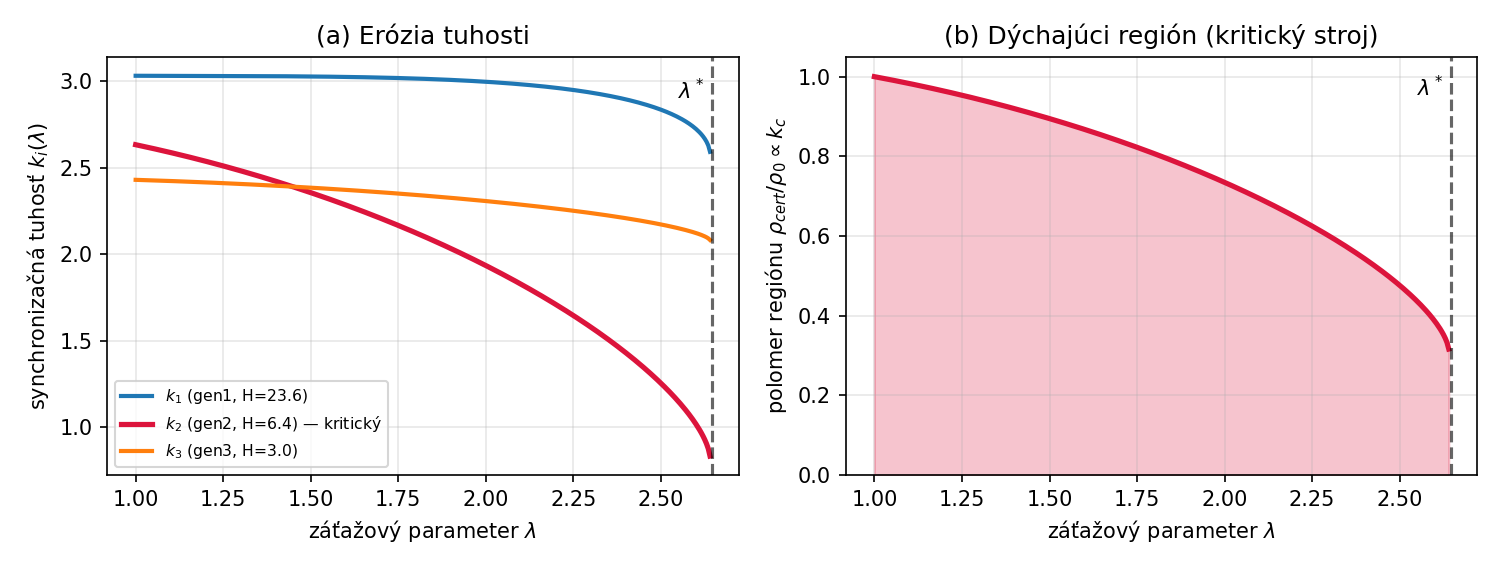}
\caption{Power-system breathing: critical-machine stiffness erosion (a) and certified region radius (b) versus loading $\lambda$, WSCC nine-bus.}\label{fig:breath}\end{figure}
\begin{figure}[tbp]\centering\includegraphics[width=\linewidth]{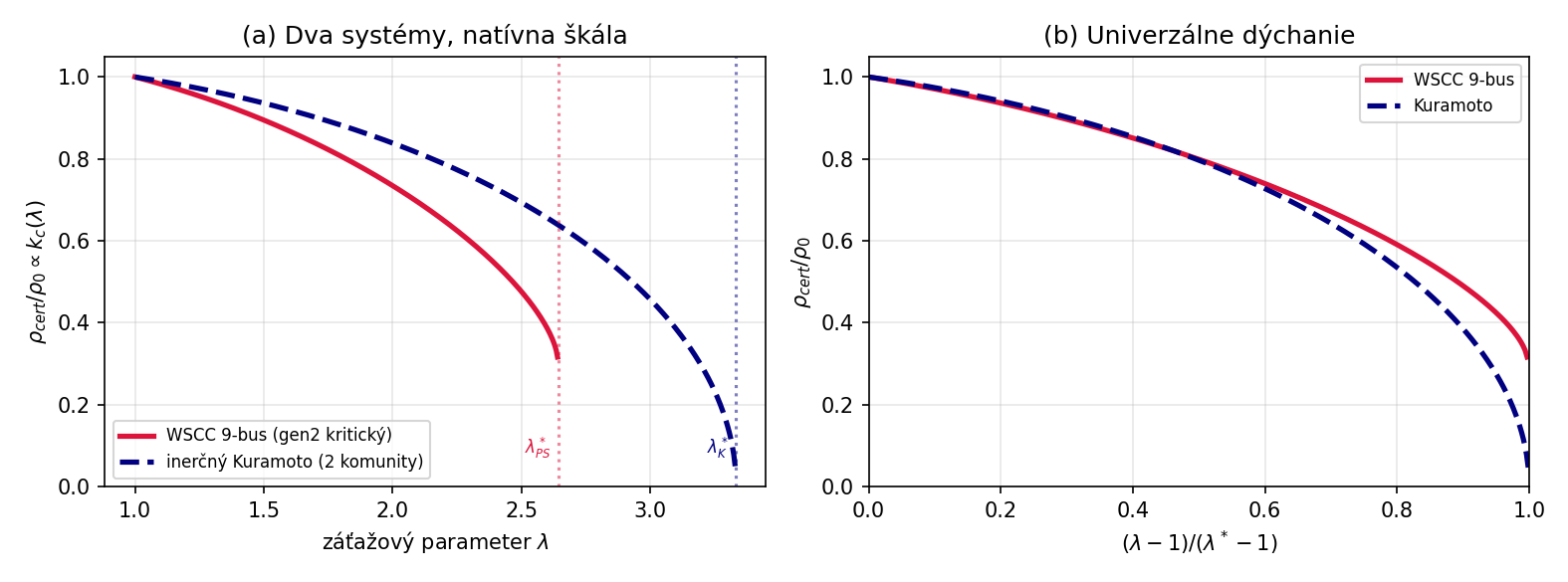}
\caption{Universality across five instances. (a) Normalised breathing curves $k_c(\lambda)/k_c(\lambda_0)$ versus distance to boundary collapse to leading order. (b) On a log--log scale against $1-s$ the curves track the universal saddle-node exponent $1/2$; the WSCC nine-bus is the mild outlier (its power-flow fold is not a pure saddle-node in these coordinates).}\label{fig:univ}\end{figure}
\begin{figure}[tbp]\centering\includegraphics[width=\linewidth]{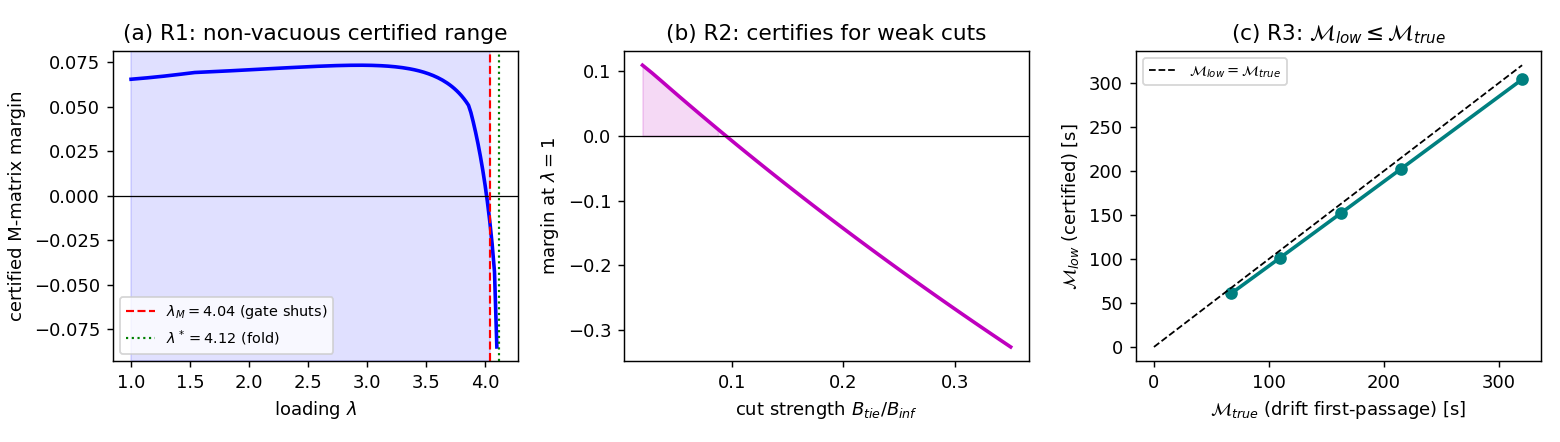}
\caption{Certifying weak-cut instance (two machines anchored to an infinite bus, weak tie). (a) The composite test matrix is a nonsingular M-matrix on $\lambda\in[1,4.04)$, shutting before the fold $\lambda^*=4.12$. (b) Certification holds for weak cuts $B_{\mathrm{tie}}/B_{\mathrm{inf}}\lesssim0.09$. (c) Under slow drift the certified margin lower-bounds the simulated first-passage, $\Mm_{\mathrm{low}}\le\Mm_{\mathrm{true}}$.}\label{fig:cert}\end{figure}

\section{Generality: rigour and the remaining step}\label{sec:gen}
The two-instance collapse (Fig.~\ref{fig:univ}b) is evidence that Theorem~\ref{thm:main} is not a power-systems result in disguise: the breathing is governed by a critical synchronising stiffness whose vanishing is a fold, and the fold is the same saddle-node normal form in both domains [44,45] (rigour level~2). The remaining step to level~1 is a proof that the normalised breathing profile is itself universal -- a centre-manifold reduction to the saddle-node normal form near the critical-stiffness fold [46], in the fast--slow setting of [32], showing $\rho_{\mathrm{cert}}/\rho_0=\Phi((\lambda-1)/(\lambda^*-1))$ for a single $\Phi$. We state this as the precise open problem, not a claim.

\section{Discussion and conclusion}\label{sec:disc}
We have developed a certified lower bound on the time-to-boundary margin $\Mm$ via a composite Lyapunov function and an M-matrix condition, identified the critical synchronising stiffness $k_c(\lambda)$ as the breathing driver, and given a constructive control dichotomy and a domain-neutral resilience--fragility reading. Three limits bound the result. The certificate is sufficient, so the M-matrix gate is conservative -- a feature (safe side), not to be sold as tight; the breathing driver $k_c$ is the tight quantity. The analysis assumes quasi-static drift ($\rho T_{\mathrm{sw}}\ll1$), inheriting the separation of [1]; under finite drift and damping the first-passage can \emph{exceed} the quasi-static fold by dynamic ride-through, so static-boundary references do not upper-bound $\Mm$ and only the certified \emph{lower} bound is guaranteed. The novelty is narrow and stated so: vector-Lyapunov M-matrix certificates are classical [2]; new are the $\lambda$-dependent test matrix tied to $\Mm$, the nonlinear double-squeeze (Proposition~\ref{prop:sector}), the breathing law $\kappa\propto k_c$, the control dichotomy (Corollary~\ref{cor:control}), and the domain-neutral universality of Section~\ref{sec:num}. Positioned within the operational-resilience programme, which calls for resilience to be quantified rather than described [47,48], $\Mm$ contributes a dynamical, certified metric: it measures, in seconds and with a guaranteed sign, the time a drifting operating point retains before the synchronism boundary.

\bmhead{Declarations}

\noindent\textbf{Funding.} No specific grant funding was received for this work.\par
\smallskip
\noindent\textbf{Competing interests.} The author declares no competing interests.\par
\smallskip
\noindent\textbf{Data availability.} The WSCC nine-bus benchmark analysed in this study is available in the cited references~\cite{andersonfouad,matpower}. The Python code reproducing all numerical results---WSCC nine-bus power flow and Kron reduction, the critical-stiffness sweep, the inertial Kuramoto instance, and the figures---together with the inertial Kuramoto network configuration, is available from the corresponding author on reasonable request.\par
\smallskip
\noindent\textbf{Author contributions.} M. Me\v{s}ter is the sole author: conceptualisation, methodology, formal analysis, computations, and writing.\par

\end{document}